\documentclass[conference]{IEEEtran}
\IEEEoverridecommandlockouts
\usepackage{cite}
\usepackage{amsmath,amssymb,amsfonts}
\usepackage{algorithmic}
\usepackage{graphicx}
\usepackage{subcaption}
\usepackage{amsthm}
\newtheorem{theorem}{Theorem}
\newtheorem{lemma}{Lemma}
\usepackage{textcomp}
\usepackage{xcolor}
\usepackage{algorithmic}
\usepackage{algorithm}
\def\BibTeX{{\rm B\kern-.05em{\sc i\kern-.025em b}\kern-.08em
    T\kern-.1667em\lower.7ex\hbox{E}\kern-.125emX}}
\begin{document}


\title{GELATO: Generative Entropy- and Lyapunov-based Adaptive Token Offloading for Device-Edge Speculative LLM Inference}


\author{Zengzipeng Tang, Yuxuan Sun, Wei Chen, Jianwen Ding, and Bo Ai\\
School of Electronic and Information Engineering, Beijing Jiaotong University, Beijing 100044, China\\
Email: \{zzptang, yxsun, weich, jwding, boai\}@bjtu.edu.cn}



\maketitle

\begin{abstract}
The recent growth of on-device Large Language Model (LLM) inference has driven significant interest in device-edge collaborative LLM inference. As a promising architecture, Speculative Decoding (SD) is increasingly adopted where a lightweight draft model rapidly generates candidate tokens to be verified by a powerful target model.
However, a fundamental challenge lies in achieving per-token resource scheduling to effectively adapt SD paradigm to resource-constrained edge environment. This paper proposes a Generative Entropy- and Lyapunov-based Adaptive Token Offloading framework, named GELATO, to maximize decoding throughput under energy constraints in a device-edge collaborative SD system. 
Specifically, an outer drift-plus-penalty loop makes online decisions to establish a reference drafting budget, managing long-term energy-throughput trade-off. Further, a nested entropy-driven generation mechanism executes early exiting to adapt to per-token dynamic generative uncertainty. 
Theoretical analysis establishes a rigorous performance bound on long-term throughput for GELATO. Extensive evaluations demonstrate that GELATO achieves a globally optimal tradeoff, outperforming state-of-the-art distributed SD architectures by 64.98\% in token throughput and reducing energy consumption by 47.47\% under resource-constrained environments, while preserving LLM decoding quality.
\end{abstract}



\section{Introduction}
While Large Language Models (LLMs) have demonstrated exceptional capabilities across diverse complex tasks, their inherent auto-regressive generation mechanism imposes significant computational demands. To alleviate this burden and reduce reliance on centralized resources,
deploying LLM at the network edge has become a strategic evolution in system architecture\cite{push}. Nevertheless, deployment on end devices remains highly impractical. Severe constraints in storage and memory I/O bandwidth strictly limit devices to hosting models with sub-10B parameters\cite{10B,edgellm}. Furthermore, restricted computational power and limited battery capacity cause devices to struggle in meeting stringent low-latency inference constraints independently\cite{latency}. Consequently, device-edge collaborative LLM inference  becomes imperative.

However, in resource-constrained edge environments, traditional device-edge split inference\cite{edge1,edge2} is fundamentally incompatible with the autoregressive mechanism, as it requires per-token transmission of large KV cache tensors, resulting in prohibitive communication overhead and latency.
To address this communication bottleneck, \textit{device-edge collaborative speculative decoding (SD)} \cite{SD1,SD2} has emerged as a transformative paradigm. By deploying lightweight small language models (SLMs) on end devices for serial drafting and a powerful LLM on the edge server for parallel verification, this architecture successfully shifts the communication payload from dense neural activations to lightweight token sequences, theoretically unlocking the potential of edge LLM\cite{push}. 

Despite the potential of device-edge collaborative speculative decoding, the limited and fluctuating edge resource availability remains a critical system bottleneck \cite{resource}. 
To enable efficient device-edge collaborative speculative decoding, it is necessary to jointly consider system resource management and autoregressive generation efficiency.
Existing studies have explored resource-aware scheduling and optimization to reduce latency and energy consumption. Specifically, work \cite{work1} proposed a resource-aware parallel speculative decoding framework, while papers \cite{work2} and  \cite{work3} developed optimization models for joint latency and energy minimization. 
Other prior work \cite{dssd} developed a framework to decouple the drafting and verification phases across different computational nodes, while work \cite{work4} and \cite{work5} utilized uncertainty-aware inference and early exits mechanisms to improve the efficiency of the speculative process.
However, these approaches typically lack the visibility required for token-level adapted resource management, leading to inefficient resource usage and performance degradation in edge environments.

To enable efficient device-edge collaborative speculative decoding under coupled resource constraints and generation dynamics, we propose a novel Generative Entropy and Lyapunov-based Adaptive Token Offloading framework, termed GELATO, where a resource-constrained device collaborates with an edge server. 

The main contributions are summarized as follows:

\begin{itemize}

    \item We propose a novel hierarchical optimization framework GELATO to jointly optimize long-term generation throughput and energy consumption in resource-constrained device-edge speculative decoding, which coordinates resource scheduling with per-token adaptation by employing a \textit{nested Lyapunov-uncertainty} method.
    

    \item To resolve the challenge of stochastic and non-explicit system token throughput optimization objective, we introduce an \textit{expectation-based surrogate} that enables step-wise Lyapunov-based online draft budget scheduling, while providing \textit{theoretical performance guarantees} for system stability and resource efficiency.

    \item We design a nested real-time early-exit mechanism for the SLM drafting stage to capture per-token generative dynamics. We establish a theoretical relationship between token acceptance rate with generative entropy, and leverage it for an entropy-driven adaptation mechanism that tracks
    the step-wise budget to ensure consistency. 

    \item Extensive evaluations are conducted across varying wireless bandwidths and LLM verification latencies. Experimental results demonstrate that our proposed approach outperforms state-of-the-art distributed split architectures, achieving a superior trade-off between throughput and energy consumption under constrained conditions.
\end{itemize}

\section{System model}
We consider a device-edge collaborative LLM system, as illustrated in Fig.\ref{fig:sys} where a resource-limited device cooperates with an edge server. To facilitate efficient inference, we deploy a lightweight SLM on device and a computation-intensive LLM on server. We assume that both models share a common vocabulary $\mathcal{V}$, which defines the set of all possible tokens. 

The system operates in a step-wise manner, where each step $k \in \{1,2,\cdots,K\}$ corresponds to a complete round of speculative decoding. At each step $k$, the SLM autoregressively generates a sequence of candidate tokens, which are then transmitted over a wireless channel to the edge server for parallel verification by the LLM.

\subsection{Speculative Decoding Model} \label{subsd}
At each decoding step $k$, the SLM generates a sequence of draft tokens $\tilde{Y}_k = [\tilde{y}_{1}, \tilde{y}_{2}, \dots, \tilde{y}_{\gamma_k}]$ with length $\gamma_k$. The computation complexity of the SLM depends on its architectural parameters. Specifically, let $N$ denote the number of transformer layers, $d_m$ the hidden dimension, and $d_f$ the dimension of the feed-forward network (FFN). Then, following \cite{work3}, the total Floating Point Operations (FLOPs) required to generate $\gamma_k$ tokens can be modeled as:
\begin{equation}
    \setlength\abovedisplayskip{5pt}
    \setlength\belowdisplayskip{5pt}
    F_k\! =\! N  \gamma_k \!\left[ 6d_m^2 + 4\left(L_k + \frac{\gamma_k}{2}\right)d_m + 2d_m^2 + 4d_m d_f \right],
\end{equation}
where $L_k$ denotes the accumulated context length by step $k$ and the term $L_k + \gamma_k/2$ accounts for the average context length during the auto-regressive drafting process. Accordingly, the inference latency of the SLM is denoted as $T^D_k = \frac{F_k}{f^D}$,
where $f^D$ denotes the computation ability of device in FLOPS.

Following the methodology in the ML.ENERGY benchmark \cite{ml}, we characterize the energy consumption of the auto-regressive decoding phase based on its steady-state power. The steady state refers to the period where the system operates under a stable and saturated workload, reflecting  long-term energy utilization. Therefore, for each step $k$, the computational energy cost at the mobile device is defined as $E^D_k = p^D T^D_k$,
where $p^D$ represents the steady-state power of SLM decoding.

Upon receiving the draft sequence $\tilde{Y}_k$, the edge server utilizes LLM to evaluate the proposals. Unlike the auto-regressive drafting phase, LLM leverages the causal masking mechanism of Transformers to verify all $\gamma_k$ draft tokens simultaneously in a single forward pass. Let $\hat{y}_i = \arg\max P(y | X_{<i})$ denote the target prediction of LLM given context $X_{<i}$, then following the standard paradigm of speculative decoding, a draft token $\tilde{y}_i$ is accepted if it matches target $\hat{y}_i$. Verification ends at the first mismatch. Edge server then appends one additional ``bonus'' token, which is either a corrected token upon mismatch or the next predicted token if the entire draft is accepted. 

We define $N_k$ as the total number of tokens hit and appended to the sequence in the $k$-th step. This value is a stochastic variable determined by draft length $\gamma_k$ and per-token acceptance outcomes of target LLM.

To quantify the stochasticity in the drafting process, we employ generative entropy as a measurable indicator of uncertainty. Specifically, for the $i$-th draft token, let $q_i(x)$ represent the predicted probability of token $x \in \mathcal{V}$. Then, the generative uncertainty quantified by the generative entropy is defined as: $H_i = - \sum_{x \in \mathcal{V}} q_i(x) \log q_i(x)$. The wall-clock latency for this verification step is considered a stable value $T^\text{LLM}_k$ given the consistent computational capability of the server \cite{tllm}.





\subsection{Probabilistic Payload-Compressed Communication Model} \label{comm}

\begin{figure}[!t]
    \centering
    \includegraphics[width=0.96\columnwidth]{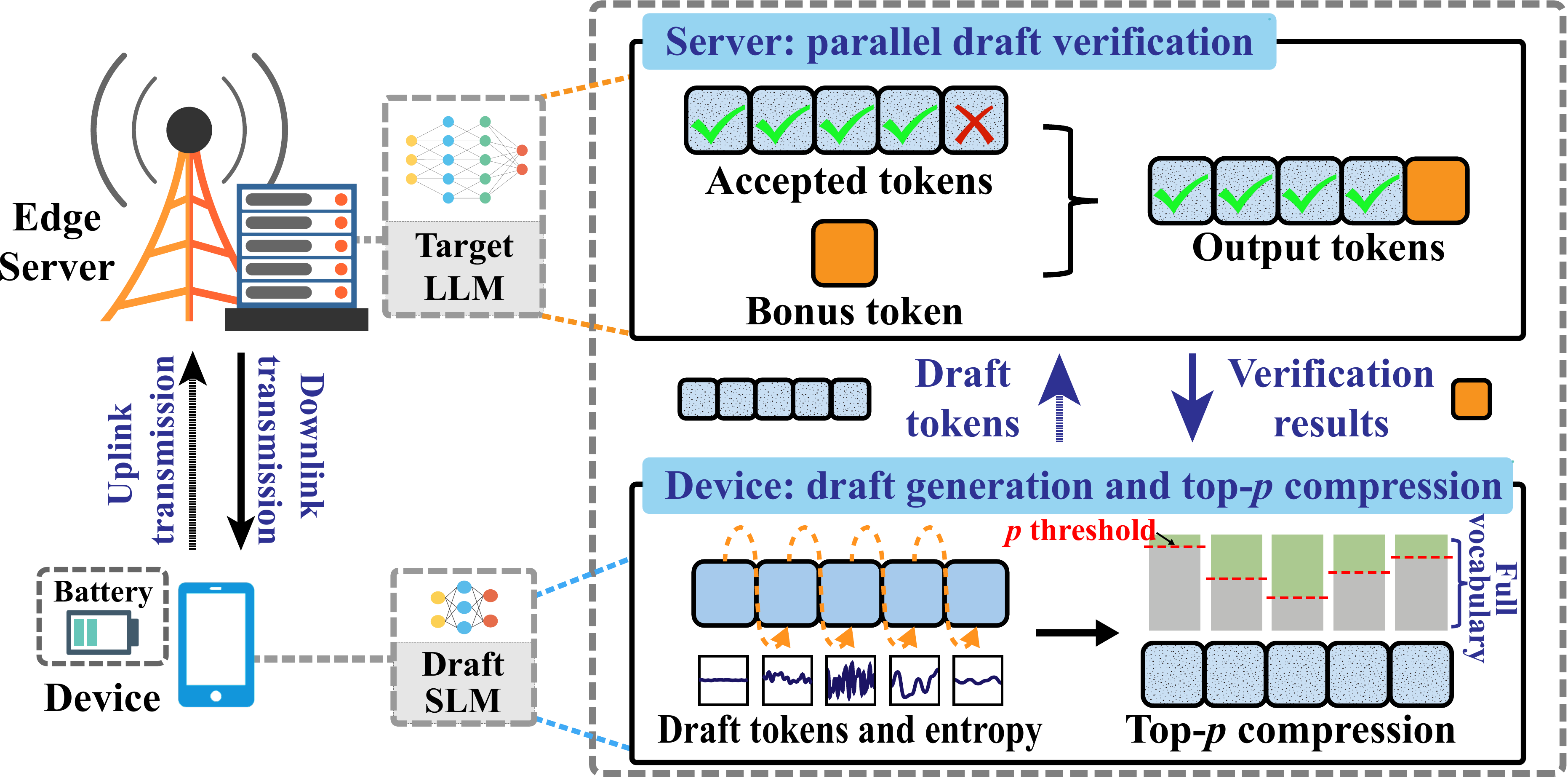} 
    \caption{Illustration of device-edge collaborative speculative decoding system.}
    \vspace{-0.4cm}
    \label{fig:sys}
    \vspace{-0pt}
\end{figure}

In conventional distributed speculative decoding \cite{work1}, device transmits full vocabulary distribution for each draft token, resulting in a communication cost proportional to $|\mathcal{V}|$. This incurs prohibitive uplink overhead for modern LLMs with $|\mathcal{V}|>10^5$\cite{dssd}. Consequently, limited uplink bandwidth and energy budgets in device-edge systems make it infeasible to transmit a large, fixed number of draft tokens at each step.

To alleviate this, we employ a probabilistic payload compression scheme based on top-$p$ truncation. Specifically, for the $i$-th drafted token, the SLM sorts the vocabulary in descending order of generation probability $q_i(x)$, and selects the smallest subset of tokens $\mathcal{S}_i \subset \mathcal{V}$ such that
$
\sum_{x \in \mathcal{S}_i} q_i(x) \ge p,
$
where $p \in (0,1)$ is a predefined probability threshold. 

By transmitting only the elements in $\mathcal{S}_i$, the uplink data volume for a draft sequence of length $\gamma_k$ can be reduced to:
\begin{equation}
    \setlength\abovedisplayskip{3pt}
    \setlength\belowdisplayskip{3pt}
    D_k = \sum_{i=1}^{\gamma_k} |\mathcal{S}_i| \left(b_\text{prob}+b_\text{index}\right), \label{payload}
\end{equation}
where $b_\text{prob}$ and $b_\text{index}$ denote the bit-width of a probability value and a vocabulary index, respectively.

The instantaneous uplink transmission rate at the decision step $k$ is modeled by the Shannon capacity formula: $r_k^U = B_k \log_2\left(1 + \frac{p^U h_k}{N_0B_k}\right)$,
where $B_k$ is the channel bandwidth, $p^U$ is the uplink transmission power, $h_k$ represents the channel gain, and $N_0$ is the Gaussian noise power spectral density.

Based on the compressed payload, the uplink communication delay is calculated as $T^U_k = \frac{D_k}{r_k^U}$.
Accordingly, the communication energy consumption is $E^U_k = p^U  T^U_k$.
Since the edge server is generally grid-powered, its computation energy is not considered. Therefore, the total energy consumption on the mobile device at decision step $k$ is :
\begin{equation}
    \setlength\abovedisplayskip{5pt}
    \setlength\belowdisplayskip{5pt}
    E_k = E_k^D+E_k^U.
\end{equation}

Meanwhile, the total latency to complete decoding step $k$ comprises the on-device drafting, uplink transmission, edge verification, and downlink feedback delays. Given that the downlink payload is minimal, its delay can be neglected. Thus, the total latency for step $k$ is given by:
\begin{equation}
    \setlength\abovedisplayskip{5pt}
    \setlength\belowdisplayskip{5pt}
    T_k = T_k^D+T_k^U+T_k^\text{LLM}.
\end{equation}



\subsection{Problem Formulation}
For step $k$, the decoding throughput is defined as $\Gamma_{k} =  \frac{N_k }{T_k}$.
Our objective is to maximize the long-term time-averaged decoding throughput while ensuring that long-term average local energy consumption of the device remains below a prescribed threshold $\bar{E}$:
\begin{subequations}
\setlength\abovedisplayskip{5pt}
\setlength\belowdisplayskip{5pt}
\begin{align}
\text{P1:} \quad 
\max_{\gamma_k} \quad 
& \frac{1}{K}\sum_{k=1}^K\Gamma_{k} \label{a} \\
\text{s.t.} \quad 
& \frac{1}{K} \sum_{k=1}^K E_k \le \bar{E} , \label{b} \\
& \gamma_k \in \{1, 2, \dots, \gamma_0\}, \quad \forall k, \label{c}
\end{align}
\end{subequations}
where (\ref{a}) is the optimization goal of maximizing the long-term decoding throughput, (\ref{b}) is the long-term energy constraint ensuring that energy of device remains within a stable range, and (\ref{c}) is used to limit the range of draft length $\gamma_k$. 

Problem P1 is essentially a \textit{stochastic optimization problem}. Its intractability arises from several coupled challenges:
(1) \textit{Unknown future dynamics}: the system targets long-term throughput and energy stability, but only instantaneous states are observable, rendering offline optimization intractable.
(2) \textit{Non-explicit formulation of $N_k$}: the objective depends on the hit token
 number $N_k$, which has no closed-form expression in terms of 
$\gamma_k$, making direct throughput optimization intractable.
(3) \textit{Limited adaptation to token generation dynamics}: step-wise scheduling operates at the granularity of a speculation step, failing to capture fine-grained, token-level dynamics, leading to inefficient resource utilization.

\section{Proposed Algorithm}
In this section, we propose the GELATO framework to solve P1. 
To resolve the intractability arising from $N_k$, we utilize an expectation-based surrogate and Lyapunov optimization to enable P1 into tractable per-step online scheduling for token drafting budgets. Further, a nested tracking mechanism is incorporated to capture dynamic token-level uncertainty via generative entropy, enabling adaptive early exiting to ensure system consistency and improve resource efficiency.

\subsection{Problem Conversion and Step-Wise Scheduling} \label{suba}
To address the intractability of the stochastic hit token number $N_k$, we utilize an expectation-based surrogate model to approximate the relationship between the draft length $\gamma_k$ and the average acceptance rate $\rho$\cite{SD1}, expressed as:
\begin{equation}
    \setlength\abovedisplayskip{3pt}
    \setlength\belowdisplayskip{3pt}
    \mathbb{E}[N_k] = 1 + \sum_{i=1}^{\gamma_k} \rho^i = \frac{1 - \rho^{\gamma_k + 1}}{1 - \rho}.
\end{equation}

Building upon this explicit surrogate, we develop a step-level online dynamic framework based on Lyapunov optimization to address the stochastic optimization problem P1. Under the drift-plus-penalty framework \cite{40,41}, long-term stochastic optimization is reduced to an online decision process that balances virtual queue stability and penalty minimization.

To handle the long-term energy constraint in (\ref{b}), we construct a virtual queue $Q_k \geq 0$ for the device, where the per-frame energy consumption $E_k$ is treated as the arrival process and the energy budget $\bar{E}$ serves as a constant service rate. The dynamics of queue backlog are given by:
\begin{equation}
\setlength\abovedisplayskip{5pt}
    \setlength\belowdisplayskip{5pt}
Q_{k+1} = \left[ Q_k + E_k - \bar{E}, 0     \right]^+. \label{QD}
\end{equation}
where $[x]^+$ is defined as $\max\{x,0\}$, and the queue is initialized with $Q_1=0$.
Then, we define the Lyapunov function for $Q_k$ as $L(Q_k) = \frac{1}{2}Q_k^2$. Based on this, the conditional Lyapunov drift is denoted as $\Delta(Q_k)=\mathbb{E}\left[L(Q_{k+1})-L(Q_k)|Q_k\right]$.


\begin{lemma}
    Using Lyapunov drift-plus-penalty framework, problem P1 can be transformed into an online, per-step optimizaiton problem with weight factor $V$:
    \begin{align}
        \setlength\abovedisplayskip{5pt}
        \setlength\belowdisplayskip{5pt}
        \text{P1.1:} \quad 
        \max_{\gamma_k} \quad 
        & U(\gamma_k)=V\Gamma_k-Q_kE_k \nonumber \\
        \text{s.t.} \quad 
        & (\ref{c}). 
    \end{align}
    where $U(\cdot)$ is a step-wise utility function.
\end{lemma}

\begin{proof}

    
    Following the classic Lyapunov optimization theorem\cite{40}, we obtain
    the upper bound of the Lyapunov drift:
    \begin{equation}
        \setlength\abovedisplayskip{5pt}
        \setlength\belowdisplayskip{5pt}
        \Delta(Q_k) \le \theta_0+Q_k\left(E_k-\bar{E}\right). \label{21}
    \end{equation}
    where $\theta_0=\frac{1}{2}\left(\max_k\{|E_k-\bar{E}|\}\right)^2$.
    Then, by incorporating the system objective as a penalty term, problem P1 can be equivalently transformed into problem (P1.1).
\end{proof}

Fundamentally, P1.1 is a non-linear integer programming problem. The objective couples a fractional structure with an exponential geometric progression. Conventional optimization techniques would lead to a highly complex transcendental equation, rendering a closed-form solution intractable.


Since $\gamma_k$ is discrete and bounded by $\gamma_0$, the optimization is one-dimensional with a limited search space. We thus obtain the optimal budget $\tilde{\gamma}_k^*$ by exhaustively evaluating $U(\gamma_k)$ over all feasible values. As shown in Algorithm \ref{algorithm1}, this guarantees optimality with a time complexity of $\mathcal{O}(\gamma_0)$.

\subsection{Dynamic Drafting via Entropy-Guided Early Exits}
Although \ref{suba} allocates the expected budget $\tilde\gamma_k^*$ to ensure long-term stability, this step-level decision remains blind to the real-time quality of the drafting process. To address this, we seek a fine-grained mechanism to monitor the generative state at a per-token granularity.



We identify a measurable indicator that reflects the likelihood of token acceptance based on the stochasticity of the drafting process and utilize it to adaptively refine the allocated budget. Specifically, we establish a consistent negative correlation between the draft token entropy $H$ and the actual target acceptance rate $\rho$ through a general function:
\begin{equation}
\setlength\abovedisplayskip{4pt}
    \setlength\belowdisplayskip{4pt}
    \rho = \varphi(H), \label{H}
\end{equation}
where $\varphi(\cdot)$ is monotonically decreasing. 

In practical deployments, each specific draft-target model pair exhibits a nominal average acceptance rate $\rho_0$, which can be obtained through offline experiments. To establish a theoretically grounded early exiting criterion, we map this statistical baseline in (\ref{H}) to an allowable uncertainty threshold:
\begin{equation}
\setlength\abovedisplayskip{5pt}
    \setlength\belowdisplayskip{5pt}
    H_\text{th} = \varphi^{-1}(\rho_0).
\end{equation}

At step $k$, SLM tracks the cumulative entropy of draft sequence and terminates drafting by maintaining an uncertainty queue to continuously monitor the cumulative generative risk. Let $\Theta_i$ denote the uncertainty backlog after the $i$-th draft token. Inspired by the leaky bucket mechanism, the queue evolves as:
\begin{equation}
\setlength\abovedisplayskip{5pt}
    \setlength\belowdisplayskip{5pt}
\Theta_i = \max \{0, \Theta_{i-1} + H_i - H_\text{th}\},  \label{hbacklog}
\end{equation}
where $H_i$ is the contextual entropy of the $i$-th token. The auto-regressive drafting immediately terminates once the backlog exceeds a predefined safety threshold $\Theta_\text{th}$. Thus, the actual transmitted draft length operating within $\tilde\gamma_k^*$ is given by:
\begin{equation}
\setlength\abovedisplayskip{5pt}
    \setlength\belowdisplayskip{5pt}
{\gamma}_k^* = \max \left\{ j \le \tilde{\gamma}_k^* \ \Big| \ \Theta_i \le \Theta_\text{th}, \ \forall i \in {1, \dots, j} \right\}. \label{best}
\end{equation}

The complete execution procedure of the proposed dual-loop framework is summarized in Algorithm \ref{algorithm1}. Note that the step-level resource queues of the outer loop are continuously updated at the end of each step using the actual execution results and real resource consumption from the inner loop.


\begin{algorithm}[t] 
\caption{ GELATO Algorithm }
\label{algorithm1}
\begin{algorithmic}[1]
\STATE Initialize the virtual energy queue $Q_1 = 0$.
\STATE \textbf{for} each decision step $k = 1, 2, \dots, K$ \textbf{do}
    \STATE \quad Observe current $h_k$ and  $Q_k$;

    \STATE \quad Initialize utility $U^* = -\infty$ and $\tilde\gamma_k^* = 1$;
    \STATE \quad \textbf{for} each candidate $i = 1$ to $\gamma_0$ \textbf{do}
        \STATE \quad\quad Evaluate utility: $U(i)$;
        \STATE \quad\quad \textbf{if} $U(i) > U^*$ \textbf{then}
            \STATE \quad\quad\quad Update optimal $U^* = U(i)$ and optimal $\tilde\gamma_k^* = i$.

    \STATE \quad Track the uncertainty backlog $\Theta$ using (\ref{hbacklog});
    \STATE \quad Determine adapted draft length $\gamma_k^*$ using (\ref{best});
    
    \STATE \quad Transmit the compressed payload corresponding to (\ref{payload});
    \STATE \quad Update $Q_{k+1}$ using (\ref{QD}).
\end{algorithmic}
\end{algorithm}

\subsection{Performance Analysis}
We now characterize the theoretical performance of the proposed GELATO algorithm by establishing a comparison with an ideal offline optimum of solving problem P1. Let $\gamma_k^\star$ represent the optimal offline decision, and let $\Gamma_k^\star$ denote the corresponding maximum throughput of step $k$.

In our setting, decisions rely on the step-wise utility ${U}(\gamma_k)$ and the corresponding energy consumption derived from the empirical target acceptance rate, denoted as $\tilde{E}_k$. Define $\sum_{k=1}^K{\Gamma}_k^\ddagger$ as the cumulative throughput of the proposed algorithm, which is obtained by calculating (\ref{a}) corresponding to (\ref{best}) in each step.
Assuming temporally independent wireless channels without specific distributional assumptions, the performance bound is given in the following theorem.

\begin{figure*}[!t]
    \centering
        
        


    \begin{minipage}[t]{0.35\textwidth}
        \centering
        
        \begin{minipage}[t]{0.43\linewidth}
            \centering
            \vspace{1.5pt}
            \includegraphics[width=\linewidth]{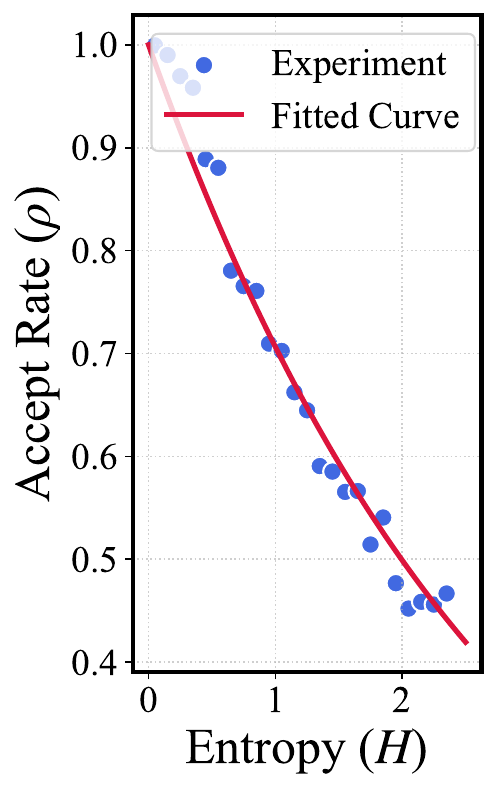}
            \captionsetup{skip=0pt}
            \captionof{figure}{The experimental and fitted curves of $H$ and $\rho$.} 
            \label{t1}
        \end{minipage}
        \hfill
        \begin{minipage}[t]{0.53\linewidth}
            \vspace{2.5pt}
            \centering
            \includegraphics[width=\linewidth]{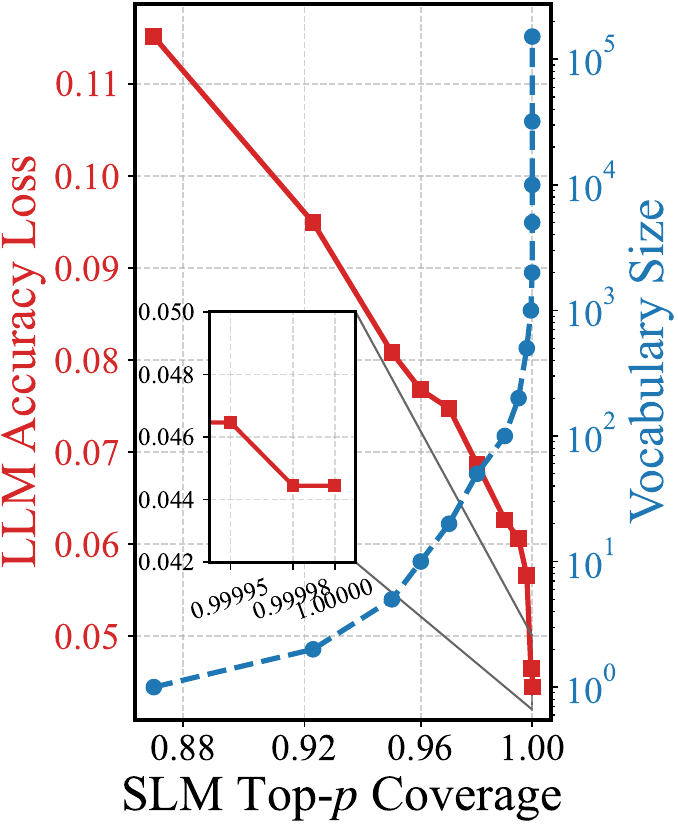}
            \captionsetup{skip=7pt}
            \captionof{figure}{LLM accuracy and vocabulary scaling under different SLM top-$p$ coverage.} 
            \label{pa}
        \end{minipage}

        \vspace{8pt} 

        \includegraphics[width=0.95\linewidth]{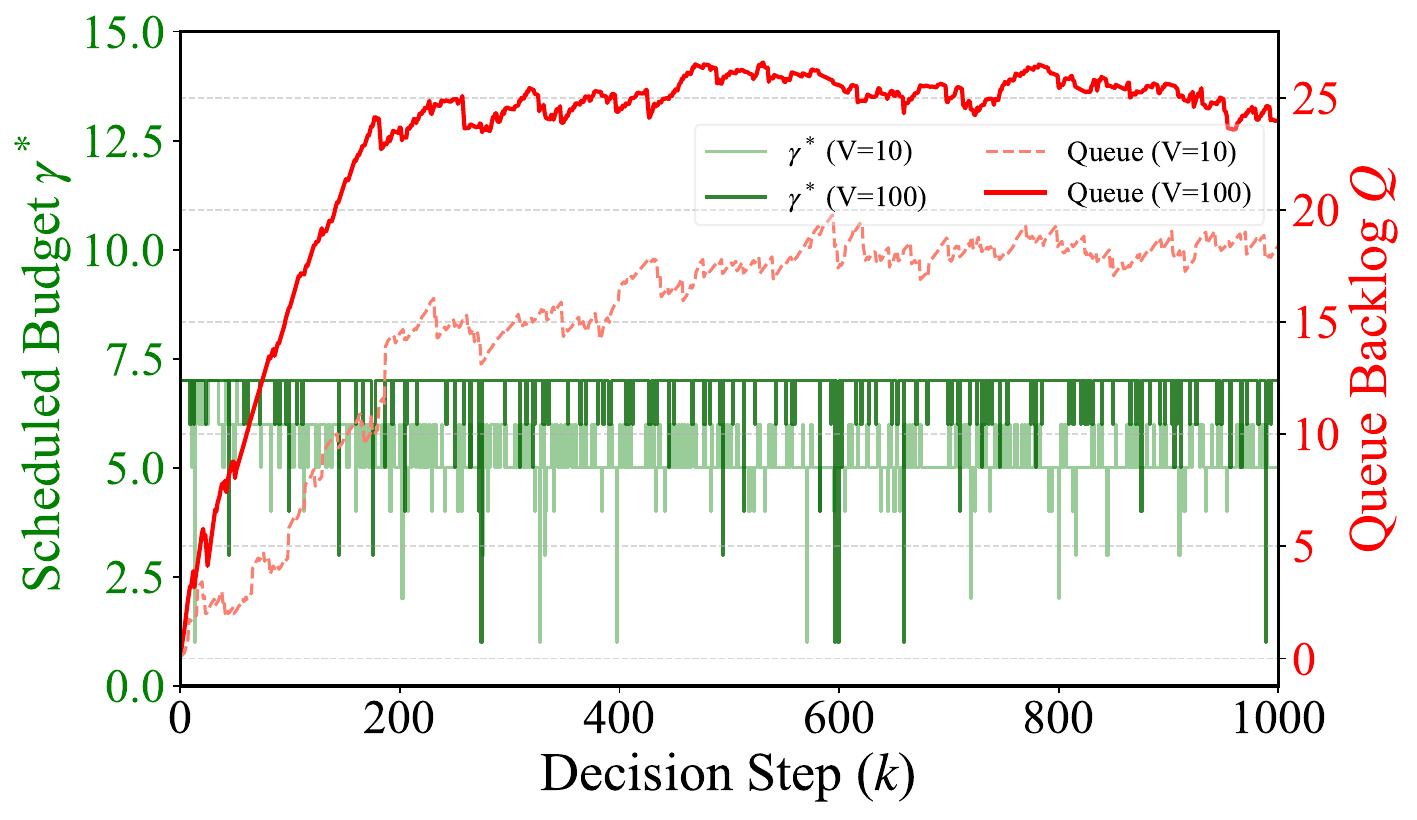}
        \captionsetup{skip=7pt}
        \captionof{figure}{Draft token budget and queue backlog under different $V$.}
        \label{t2}
        
    \end{minipage}
    \hfill
    \begin{minipage}[t]{0.62\textwidth}
        \vspace{0pt}
        \centering
        \begin{subfigure}[t]{0.47\linewidth}
            \centering
            \includegraphics[width=\linewidth]{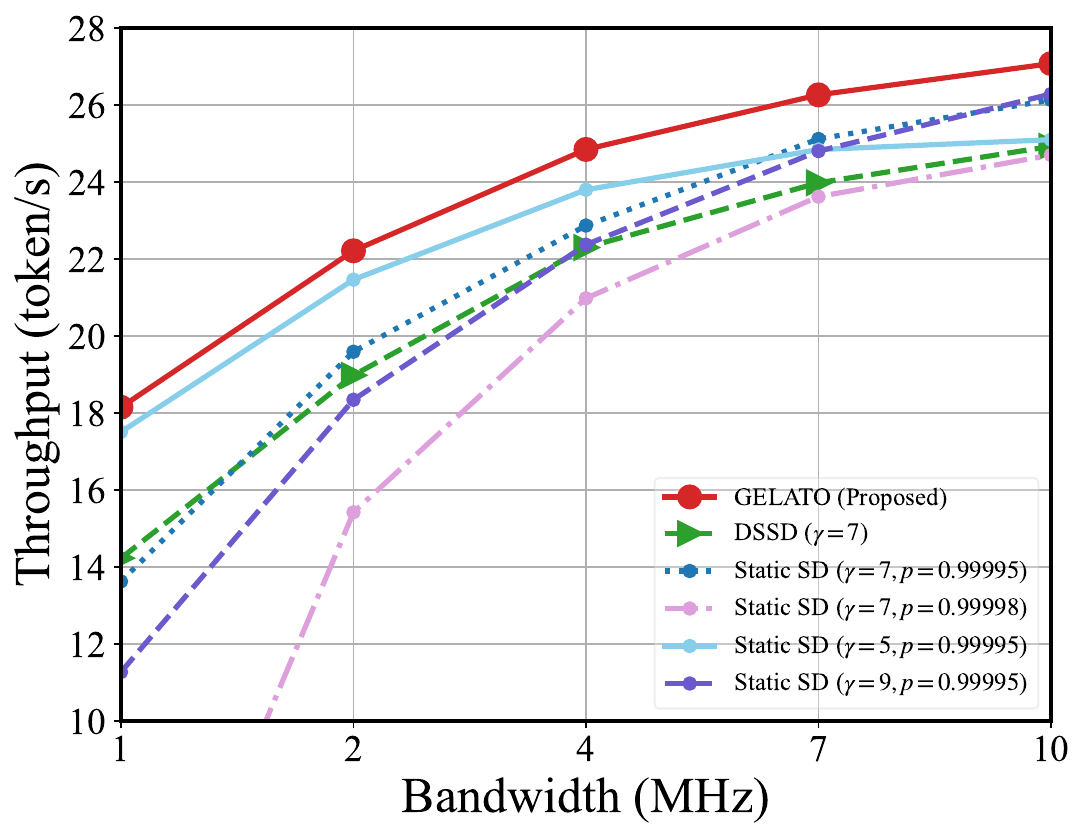}
            \caption{System throughput under different bandwidth.}
            \label{fig:sub1}
        \end{subfigure}
        \hfill
        \begin{subfigure}[t]{0.48\linewidth}
            \centering
            \includegraphics[width=\linewidth]{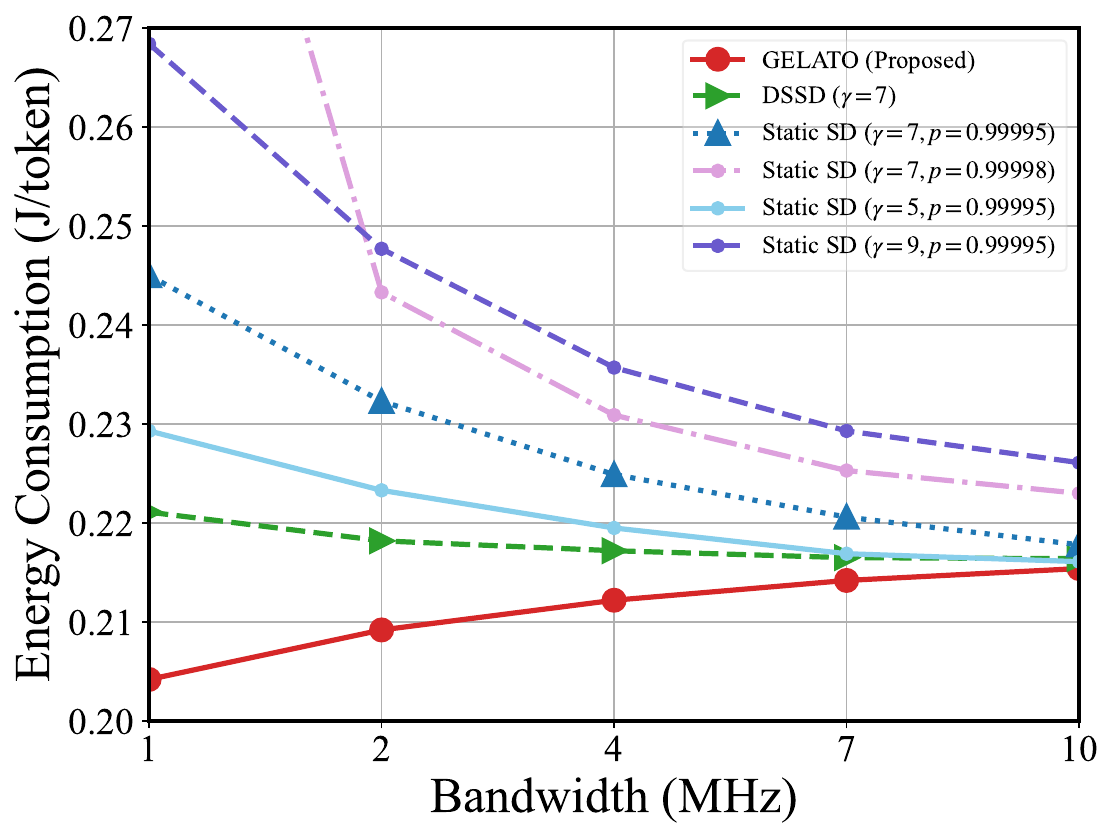}
            \caption{Energy consumption under different bandwidth.}
            \label{fig:sub2}
        \end{subfigure}

        \vspace{15pt} 

        \begin{subfigure}[t]{0.47\linewidth}
            \centering
            \includegraphics[width=\linewidth]{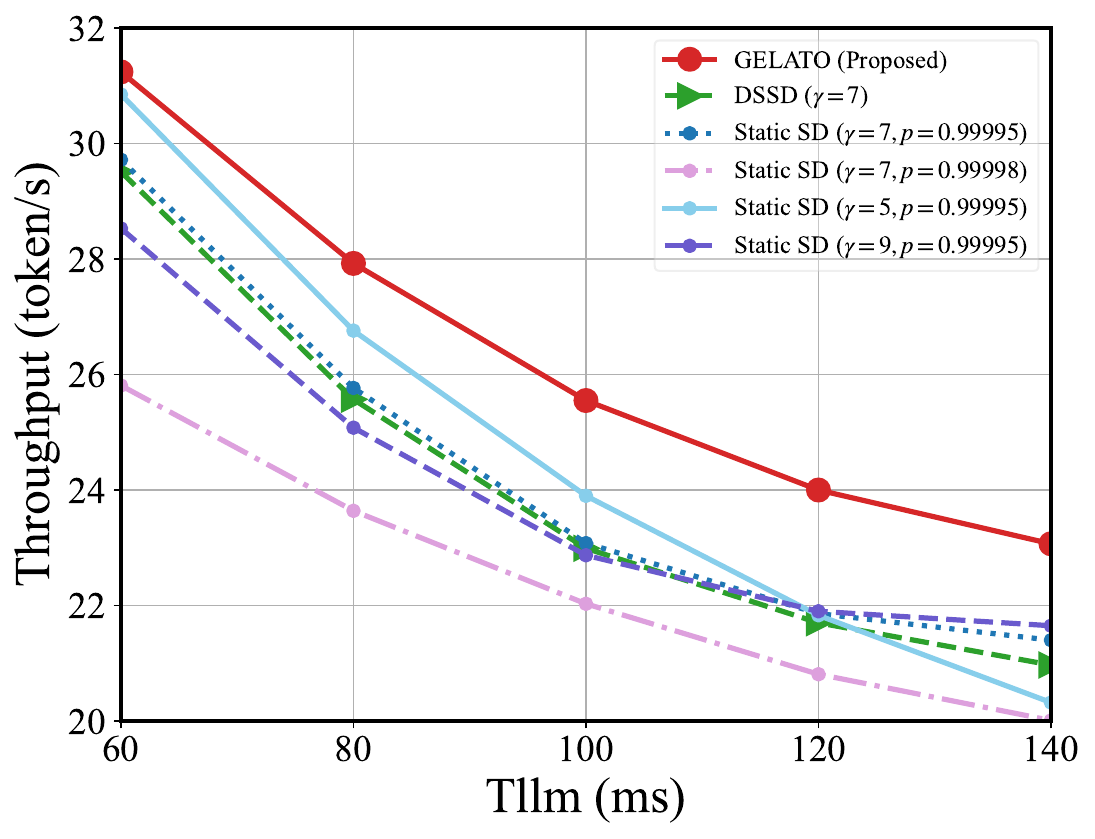}
            \caption{System throughput under different verification latencies.}
            \label{fig:1}
        \end{subfigure}
        \hfill
        \begin{subfigure}[t]{0.48\linewidth}
            \centering
            \includegraphics[width=\linewidth]{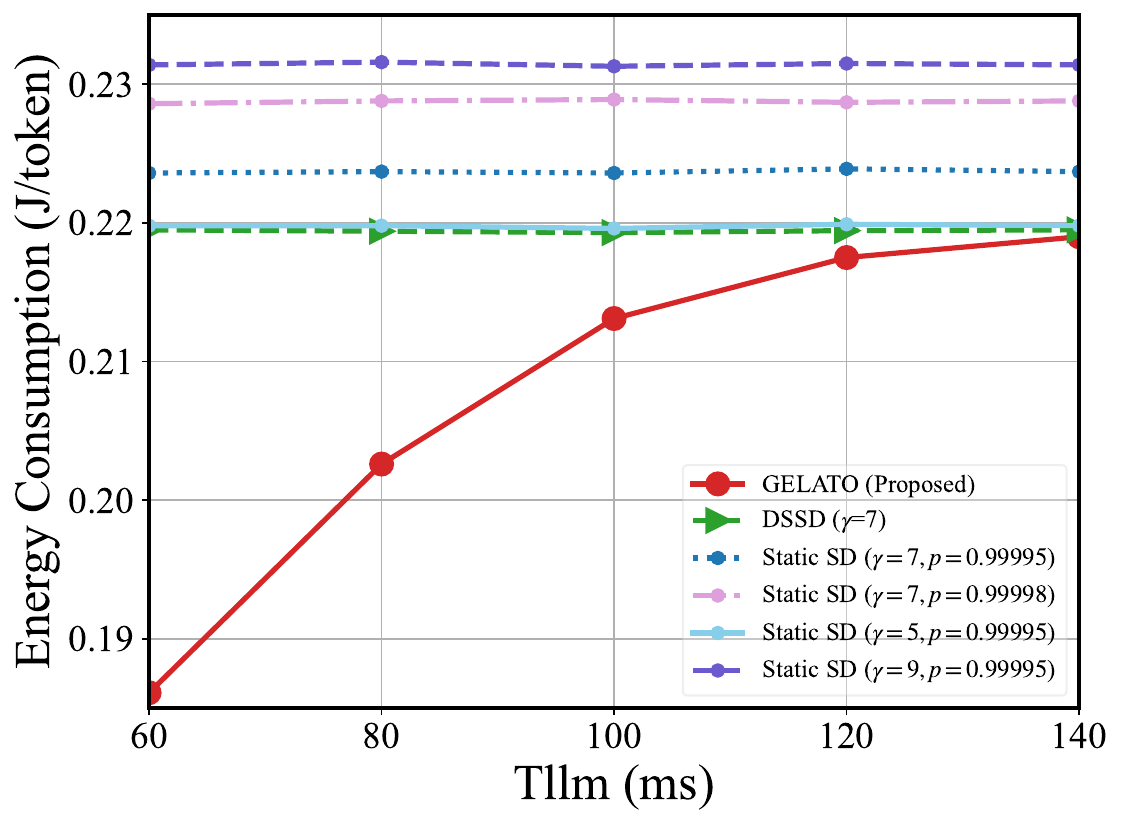}
            \caption{Energy consumption under different verification latencies.}
            \label{fig:2}
        \end{subfigure}
        
        \caption{Performance of the proposed algorithm and benchmarks on GSM8K.}
        \label{fig:main_figure}
    \end{minipage}
\end{figure*}

\begin{theorem}\label{theorem1}
    Relative to the offline optimal solution, the cumulative throughput of GELATO is bounded by
    \begin{equation}
    \setlength\abovedisplayskip{4pt}
    \setlength\belowdisplayskip{4pt}
        \sum_{k=1}^{K}\!\Gamma_{k}^{\ddagger} \ge \sum_{k=1}^{K} \!\Gamma_{k}^{\star}\!-\!\frac{\theta_0^2K^2 \!+\!K (K\!-\!1)\delta_0\theta_0}{2V}, \label{G}
    \end{equation}
    and the cumulative energy violation is bounded by
    \begin{equation}
    \setlength\abovedisplayskip{3pt}
    \setlength\belowdisplayskip{3pt}
        \sum_{k=1}^K E_{k}\le K\bar{E}+\sqrt{2\theta_0K^2 +2K(K-1)\delta_0\theta_0}. \label{e}
    \end{equation}
\end{theorem}
\begin{proof}


Define the $K$-step drift as $\Delta_{K} \triangleq L(K+1)-L(1)$ and $\delta_k=\tilde{E}_k-E_k$. Then the $K$-step drift-plus-penalty function is bounded by $\Delta_{K} -V\Gamma_k \le    
    \theta_0K+\sum_{k=1}^{K}\left(  Q_k\left(\tilde{E}_k-\bar{E}\right)-V\Gamma_k-  Q_k\delta_k\right)$.
    

We use $^\star$, $^\dagger$, and $^\ddagger$ to denote the optimal offline solution of P1, the classical drift-plus-penalty solution to P1.1, and the result of our proposed GELATO, respectively. The drift-plus-penalty over $K$ steps satisfies the following bound:
 \vspace{-5pt}
\begin{align}
\setlength\abovedisplayskip{5pt}
    \setlength\belowdisplayskip{5pt}
   &\Delta_{K}^{\dagger}\!-\!V\!\sum_{k=1}^{K}\!\Gamma_{k}^{\dagger}   \!\le \!\theta_0 K\sum_{k=1}^{K} \left(   Q_k \left(\tilde{E}_k-\bar{E}\right)^\ddagger - V \Gamma_k^\ddagger -   Q_k \delta_k^\ddagger \right) \nonumber\\
    & \!\!\stackrel{(a)}{\le}\!  \theta_0 K\!+\!\sum_{k=1}^{K} \!\left(   Q_k \left(\tilde{E}_k\!-\!\bar{E}\right)^\dagger\! - \!V {\Gamma}_k^\dagger -   Q_k \delta_k^\ddagger \right) \nonumber\\
     &\!\!\stackrel{(b)}{\le}\! \theta_0 K\!+\!\sum_{k=1}^{K}\!\left( Q_k\!\left(E_k\!-\!\bar{E}_n\right)^{\star}\!-\!V\Gamma_k^{\star} \!+\!2\delta_{0} Q_k\!\right)\!,\label{8}
\end{align}
where $\delta_0\triangleq \max_{k}\left\{\left|\tilde{E}_{k}-E_{k} \right|\right\}$. Inequality (a) follows since the optimal solution to P1 minimizes the objective at each step $k$. Inequality (b) holds because the drift-plus-penalty method minimizes $Q_k E_k - V \Gamma_k$; thus, replacing it with the offline optimal policy can only yield a larger or equal value.

Now we bound the right-hand-side of (\ref{8}). Note that $Q_{k+1}-Q_{k}\le \theta_0, \forall k$, and thus $Q_k = Q_k-Q_{1}\leq(k-1)\theta_0$, $Q_k\left({E}_k-\bar{E}\right)^\star=\left( Q_k-Q_{1}\right)\left({E}_k-\bar{E}\right)^\star \le(k-1)\theta_0^2.$

Substituting them into (\ref{8}) yields:
\begin{equation}
\setlength\abovedisplayskip{4pt}
    \setlength\belowdisplayskip{4pt}
   \Delta_{K}^{\ddagger}\!-\!V\!\sum_{k=1}^{K}\!\Gamma_{k}^{\ddagger}  
   \!\le\!-V\!\sum_{k=1}^{K} \!\Gamma_{k}^{\star}\!+\!\frac{1}{2}\theta_0^2K^2 \!+\!K (K\!-\!1)\delta_0\theta_0.\label{11}
\end{equation}

Since $\Delta_{K}^{\ddagger} \!\ge\! 0$, results in (\ref{G}) follows from (\ref{11}) by dividing both sides by $V$. As $\Gamma_k \ge 0$, and $\frac{1}{2}Q^2_{K+1}\le \Delta_K$, we get:
\vspace{-8pt}
\begin{align}
    \sum_{k=1}^K \left(E_{k}-\bar{E} \right)&\le \sum_{k=1}^K \left( Q_{k+1}\!-\!Q_{k}\right)\! =\! Q_{K+1}\!\le\! \sqrt{2\Delta_K}\\&\le\sqrt{2\theta_0K^2 +2K(K-1)\delta_0\theta_0}.
\end{align}

\vspace{-6pt}
Thus, (\ref{e}) in Theorem \ref{theorem1} is proved.
\end{proof}

\section{Simulation Results}
In this section, we conduct simulations to evaluate the performance of our proposed framework. 
We simulate an edge speculative decoding system over a channel modeling both large-scale path loss and small-scale Rayleigh fading.
The inference task deploy Qwen2.5-7B-Instruct\cite{qwen2.5} as the target model on the edge server and Qwen2.5-0.5B-Instruct\cite{qwen2.5} as the local draft model. We utilize the GSM8K dataset\cite{gsm8k}, a math benchmark designed to challenge multi-step reasoning.
Through sufficient profiling on this model pair, the average acceptance rate is measured at $\rho_0=0.9$. All reported results are averaged over 1000 simulation rounds.
Unless otherwise specified, the main simulation parameters are set according to the configurations in \cite{edgeshard,work3 }, and summarized in Table \ref{tab:sim_params}.


\begin{table}[htbp]
    \centering
    \caption{Simulation Parameters}
    \label{tab:sim_params}
    \begin{tabular}{l|l}
        \hline
        \textbf{Parameter} & \textbf{Value} \\
        \hline
        Noise power spectral density $N_0$  & -174 dBm/Hz \\
        Device GPU frequency ($f^D$) & 40.0 GFLOPS \\
        Device Computation Power ($p^D$) & 12.0 W \\
        Uplink Transmit Power ($p^U$) & 23 dBm \\
        Long-term Energy Budget ($\bar{E}$) & 1.2 J \\
        Draft Length Capacity ($\gamma_0$) & 15 \\
        Uncertainty Backlog Threshold ($\Theta_\text{th}$) & 1.2$H_\text{th}$ \\
        \hline
    \end{tabular}
    \vspace{-0.2cm}
\end{table}


To validate the proposed uncertainty-acceptance rate model, we perform teacher-forced inference using the Qwen2.5 model pair. Fig. \ref{t1} shows the mapping in (\ref{H}) follows an exponential decay, well approximated by $\rho = \exp(-0.35H)$, with the uncertainties discretized into uniform bins of step size 0.1.


To demonstrate that our proposed payload compression communication method in \ref{comm} effectively preserves decoding quality, Fig.\ref{pa} evaluates the impact of SLM vocabulary coverage $p$ on LLM accuracy, which is normalized to the official Qwen baseline of 51.7\%\cite{qwen2.5} on GSM8K. The inset plot highlights that under $p=0.99995$, LLM achieves nearly lossless decoding quality compared to full-payload.

To evaluate the step-wise scheduler, we consider two penalty weights with $V=10$ and $V=100$, which balances the expected generative yield and the virtual energy deficit. A larger $V$ prioritizes throughput maximization, resulting in more aggressive drafting budget allocation, whereas a smaller one emphasizes energy stability by limiting the budget and queue backlog, as illustrated in Fig.\ref{t2}. These results validate that the drift-plus-penalty framework enables tunable adaptation to resource and performance requirements.

To better demonstrate the superiority of the proposed hierarchical framework, we consider the following benchmark schemes: 1) \textit{Static Speculative Decoding} (Static SD)\cite{SD1,SD2} where the system generates a fixed number of draft tokens $\gamma \in \{5, 7, 9\}$ per step. To ensure a rigorous baseline comparison, this scheme is specifically enhanced with a top-$p$ payload compression of $p= 0.99995$ and $p=0.99998$; 2) \textit{Distributed Split Speculative Decoding} (DSSD) \cite{dssd}, representing the state-of-the-art communication-efficient architecture. This scheme mitigates the uplink bottleneck by replacing multiple uplink distribution transmissions with a single conditional downlink transmission, with the device receive power set to 19 dBm and the maximum draft token length set to 7.

We then evaluate the system performance under varying channel bandwidth with a fixed $T_k^\text{LLM}$ of 100ms. As shown in Fig.\ref{fig:main_figure}(\subref{fig:sub1}) and Fig.\ref{fig:main_figure}(\subref{fig:sub2}), our proposed framework achieves better throughput and energy consumption trade-off across all conditions. The advantage is most pronounced in the severely communication-constrained region. Specifically at 1 MHz, our method realizes up to 64.98\% throughput gain. Meanwhile, our approach reduces energy consumption by 47.47\% at 1 MHz and maintains the lowest usage across all bandwidths. This robustness stems from the joint effect of the outer-loop dynamic budget allocation and the inner-loop entropy-driven adaptation, instead of a fixed generation and transmission scheme. DSSD also shows robustness under limited bandwidth, whereas Static SD degrades in throughput and energy efficiency due to the absence of dynamic sequence adaptation.

We also evaluate the performance under varying target verification latencies with a fixed bandwidth of 5MHz. As shown in Fig.\ref{fig:main_figure}(\subref{fig:1}) and Fig.\ref{fig:main_figure}(\subref{fig:2}), our framework achieves a throughput gain up to 20.04\% over benchmarks at the stringent 140ms verification latency, and smoothly rising to 25.95\% at 60ms. Concurrently, it reduces energy consumption by up to 23.39\%. This is also attributed to the dynamic adaptation algorithm with entropy reference, resulting in stable energy consumption while achieving robust throughput against varying LLM delays. In comparison, DSSD framework also considers payload compression and shows some robustness to changes in latency. Static SD schemes suffer from throughput and energy inefficiency under high latency due to fixed-length drafting without dynamic early exiting.

\section{Conclusion}
This paper proposed a tightly coupled dual-loop speculative decoding framework, named GELATO, for efficient device-edge collaborative speculative decoding. By nesting microscopic entropy-driven dynamic halting within macroscopic Lyapunov scheduling, the system successfully resolves endogenous acceptance dynamics and time-coupled stochasticity. Evaluations demonstrate a globally optimal tradeoff between generation throughput and energy consumption, significantly outperforming state-of-the-art distributed architectures under severe communication bottlenecks. Future work will extend this collaborative optimization to the context prefill phase via distributed prompt processing and cache synchronization.

\bibliographystyle{IEEEtran}
\bibliography{IEEEabrv,reference}

\end{document}